\documentclass[a4paper,10pt]{article}
\usepackage[utf8]{inputenc}
\usepackage{xspace}
\usepackage{subfigure}
\usepackage{amsmath}
\usepackage{enumerate}
\usepackage{framed}
\usepackage{graphicx}
\usepackage{lscape}
\usepackage{bm}
\usepackage{color}
\usepackage{multicol}
\usepackage{amsthm}
\usepackage{amssymb}

\renewcommand{\eqref}[1]{{(\ref{#1})}}

\newtheorem{theorem}{Theorem}[section]

\newtheorem{lemma}[theorem]{Lemma}

\newtheorem{definition}[theorem]{Definition}
\newtheorem{assumption}[theorem]{Assumption}
\newtheorem{remark}[theorem]{Remark}

\newcommand{\thmref}[1]{{Theorem~\ref{#1}}}
\newcommand{\lemref}[1]{{Lemma~\ref{#1}}}
\newcommand{\assref}[1]{{Assumption~\ref{#1}}}

\title{ A  STOCHASTIC DELAY MODEL FOR PRICING DEBT AND LOAN GUARANTEES: Theoretical results}

\author{Elisabeth Kemajou\footnote{Email: isakema@umn.edu, Corresponding  author,  Department of Mathematics, University of Minnesota, USA},
Salah-Eldin A. Mohammed\footnote{Email: salah@sfde.math.siu.edu, Department of Mathematics, Southern Illinois University Carbondale, USA} and Antoine Tambue\footnote{Email: Antoine.Tambue@math.uib.no, Department of Mathematics, University of Bergen, Johannes Bruns Gate 12,  5020, Bergen--Norway} }

\begin{document}

\maketitle

\begin{abstract}
We consider that the price of a firm follows a non linear stochastic delay differential equation. We also assume that any
claim value whose value depends on firm value and time follows a non linear stochastic delay differential equation. Using
self-financed strategy and replication we are able to derive a random partial differential equation (RPDE) satisfied by any corporate claim whose value is a function of firm value and time. Under specific final and boundary conditions, we solve the RPDE for the debt value and loan guarantees within a single period and homogeneous class of debt. 
\end{abstract}


\section{INTRODUCTION}
The valuation of corporate claims has always been an important topic for
finance researchers. On one hand, bond issuers would like to know what factors affect prices and
yields, as yields represent their cost of capital. On the other hand, prospective bond buyers
are interested in knowing how sensitive yields and yield spreads are to various relevant
factors (e.g. leverage) as they develop investment strategies.
Due to the significant growth of the credit derivatives market, the interest in corporate claims values
 models and risk structure has recently increased. 
This growth is explained by the need of better prediction models to fit the real market data.

Corporate bankruptcy is central to the theory of the firm. A firm is generally considered bankrupt when 
it cannot meet a current payment on a debt obligation.
In this case the equityholders lose all claims on the firm,
and the remaining loss which is the difference between the face value of the fixed claims and the market 
value of the firm, is supported by the debtholders. 
This is the definition of bankruptcy that we adopt in this paper. Loan guarantees have been proposed by 
several authors as a way to encourage new investments for companies when they become insolvent.

The risk structure of interest rates on bonds with the same maturity is degree of the likelihood of 
default on the payment of interest and the principal. 
Returns are measured by yields to maturity of each bond. The difference between the yields of bonds
 with default risk and default free bonds, the yield spreads,
is defined as the spread between their interest rates. This yield spread is sometimes called risk 
premium since it is supposed to measure the additional yield that
risky bond pay in order to motivate investors to buy risky bonds instead of the less risky ones.
It does not seem to be a consensus among the researchers what the determinants of the risk structure are.
 Different variables have been considered to represent
a valid measure of risk depending on whether the same maturity or different maturities, see (\cite{GD2}). 

The current model in corporate finance was developed by Merton \cite{GD2,GD3}. 
This model is closed to Black and Scholes model for stock price. As the Black and Scholes model \cite{scott,RU}, the fitness of Merton model
can be questioned  since the model assume that the volatility
is constant  and  empirical evidence shows that volatility actually depends on time in a way that is not predictable.
Beside, the need for better ways of understanding the behavior of many natural processes has
motivated the development of dynamic models of these processes that
take into consideration the influence of past events on the current and
future states of the system \cite{ito,Mi,Ma,MS}.
Following  a delayed Black and Scholes formula proposed in \cite{GD4} we developed the corresponding delayed model
in corporate finance, which  has not yet been introduced. Because of the 
isomorphic relationship between levered equity and a European
call option (see Merton \cite{GD7}) on one hand, and the isomorphic correspondence between loan 
guarantees and common stock put options (see Merton \cite{GD3}), 
we can claim and prove that results obtained in the theory of option pricing are feasible 
in corporate liabilities pricing. 

The paper is organized as follows. In section 2 we present the stochastic delay model for corporate claims. 
In this section we provide keys definitions in corporate finance used in this work and develop the 
Random Partial Differential Equation for claims. In section 3, we evaluate the debt in a levered firm. 
The evaluation of the loan guarantees is provided in section 4. We end by studying the 
impact of an additional debt on the firm's risk structure in section 5.
 
\section {STOCHASTIC DELAY MODEL FOR CORPORATE CLAIMS}
Let us start by providing some keys definitions in finance which  will be extensively used in this work.
 \subsection{Keys Definitions}

\begin{definition}
$\left[\textbf{Firm value or Company value}\right]$\\
 The firm value or Company value is the market value of the company's machines and commercial activities. 
This value is equal to the market value of the equityholders plus the market value of the net financial debt.
\end{definition}

\begin{definition}
$\left[\textbf{Equity Value}\right]$\\
 The equity is the total dollar market value of all of a company's outstanding shares. Market value of
equity is calculated by multiplying the company’s current stock price by its number of
outstanding shares. It’s the total value of the business after taking out the amount owed
to debtholders.
\end{definition}
\begin{definition}
$\left[\textbf{Corporate claim or corporate liability \cite{Va}}\right]$\\
A corporate claim or corporate liability is
an official request for money usually in the form of compensation, from a corporation.
\end{definition}
\begin{definition}
$\left[\textbf{Debt security \cite{Va}}\right]$\\
A Debt security is a security issued by a company or government which represents money borrowed from the
security's purchaser and which must be repaid at a specified maturity date, usually at a
specified interest rate.
\end{definition}
\begin{definition}
 $\left[\textbf{Loan Guarantees) \cite{Va}}\right]$\\
Loan on which a promise is made by a third party or guarantor that he or she will be liable
if the creditor fails to fulfill their contractual obligations.
\end{definition}
\subsection{Stochastic Model and Random Partial Differential Equation for claims}
 In order to obtain a better prediction of the company value, we need to include its history. 
 Let us assume  that the price of the firm $V(t)$ at time $t\in [0,T]$
%
follows a non linear Stochastic Delay
 Differential Equation (SDDE) of the form
\begin{eqnarray}
\label{model}
\left\lbrace \begin{array}{l}
 dV(t)=(\alpha V(t)V(t-L_1)-C)dt+g(V(t-L_2))V(t)dW(t) \\
 \newline\\
 V(t)=\varphi(t),\,\, t\in [-L,0]\\
\end{array} \right.
\end{eqnarray}
on a probability space $(\Omega,\mathcal{F},P)$ with a filtration $(\mathcal{F}_t)_{0\leq t\leq T})$ satisfying the usual conditions.
The  constants $L_1$ and $L_2$ are positive,  $\alpha$ is the  riskless interest rate of return on the firm per unit time, $C$ 
is the total amount payout by the firm per unit time to either the shareholders or claims-holders (e.g.,dividends or interest payments) 
if positive, and it is the net amount received by the firm from new financing if negative.  
The constant $L=\max(L_1,L_2)$) represents the past length while $T$ is the maturity date.
The function  $g:\mathbb{R}\rightarrow\mathbb{R}$ 
is a continuous representing the volatility function on the firm value per unit time.
The initial process $\varphi : \Omega \rightarrow C([-L,0],\mathbb{R})$ is $\mathcal{F}_0$-measurable with respect to the 
Borel $\sigma$-algebra of $C([-L,0],\mathbb{R})$, 
actually $\varphi$ is the past price of the firm.
The process $W$ is a one dimensional standard Brownian motion adapted to the filtration $(\mathcal{F}_t)_{0\leq t\leq T}$.

The following theorem ensure that  the price model  \eqref{model} is feasible in the sense that it admit pathwise unique solution such that $V(t)> 0$  almost surely for all $t \geq 0$ 
whenever the initial path $\varphi (t)> 0$ for all $t \in [0, t]$.
\begin{theorem}
\cite{GD4}\label{Lem:L30}\\
The firm price model \eqref{model} has a unique solution. Furthermore, if  $C=0$
 the solution is represented  by the formula 
\begin{eqnarray} 
V(t)=\varphi(0) \exp\left(\left(\int_{0}^{t} \alpha(s) V(s-L_1)ds-\frac{1}{2}\int_0^t(g(V(s-L_2)))^2ds\right) \right.\\
+ \left. \int_0^t g(V(s-L_2))dW(s)\right) \nonumber\\
\end{eqnarray}

\end{theorem}
\begin{proof}
Proof can be found in \cite{GD4}, where the authors  deal with stock price.
\end{proof}

 Following the work in \cite{GD2}, in order to derive a random partial differential equation
which must be satisfied by any security whose value can be written as a function of the
value of the firm and time. 
We assume that any claim with market value $Y(t)$ (which can be replicated using self-financed strategy) at time  $t$ with $Y(t)=F(V(t),t)$ follows
a non linear stochastic delay differential equation
\begin{eqnarray}\label{eq1}
\left\lbrace \begin{array}{l}
  dY(t)=(\alpha_y Y(t) -C_y)dt+g_y(Y(t-L_2))Y(t)dW_y(t), \,\, t\in [0,T]\\
  \newline\\
  Y(t)=\varphi_y(t),\,\, t\in [-L,0],
  \end{array} \right.
\end{eqnarray}
on a probability space $(\Omega,\mathcal{F},P)$.
The constant $\alpha_y$ is the riskless interest rate of return per unit time on this claim; $C_y$ is the amount payout 
per unit time to this claim; $g_{y}: \mathbb{R}\rightarrow \mathbb{R}$ is a continuous function representing the volatility 
function of the return on this claim per unit time; the initial process $\varphi_y : \Omega \rightarrow C([-L,0],\mathbb{R})$ 
is $\mathcal{F}_0$-measurable with respect to the Borel $\sigma$-algebra of $C([-L,0],\mathbb{R})$.
The process $W_y$ is a one dimensional standard Brownian motion adapted to the filtration $(\mathcal{F}_t)_{0\leq t\leq T}$.

\begin{assumption}\label{A2}
The value of the company is unaffected by how it is financed (the capital structure irrelevance principle).
\end{assumption} 
\begin{theorem}
\label{th1}
 Assume that the value of the firm $V(t),\,t\in[0,T]$ follows the SDDE \eqref{model}. Furthermore, suppose that \assref{A2} 
 is satisfied and that the debt value accumulates interest compounded continuously at a rate of $r$, that is $B(t)=B(0)e^{rt}$. For any claim whose value is a function of the firm value and time i.e.
$Y(t)=F(V(t),t)$, where $F$ is twice continuously differentiable with respect to $V$ and once differentiable with respect to $t$, the following RPDE should be satisfied;

\begin{eqnarray}\label{eq:3008}
\dfrac{1}{2}g^2(V(t-L_{2}))v^2F_{vv}+(rv-C)F_v+F_t-rF+C_y=0,\,(t,v)\in (0,T) \times \mathbb{R}^{+}
\end{eqnarray}
with
$$F_t(v,t)=\frac{\partial F(v,t)}{\partial t},\, F_v(v,t)=\frac{\partial F(v,t)}{\partial v},\,F_{vv}(v,t)=\frac{\partial^2 F(v,t)}{\partial v^2}.$$
\end{theorem}
\begin{proof}
The proof is closed to one in \cite{GD2} for non delayed model.
Note that, for a given $Y(t)=F(V(t),t)$, there are similarities between the $\alpha_y, g_y, dW_y$ and 
the corresponding $\alpha, g, dW$ in SDDE \eqref{model}.
Knowing that $V(t)$ is an Ito process and since we assumed that $F$ is twice continuously differentiable with respect to $v$ and once differentiable with respect to $t$, 
Ito formula, allows us to write the following for $Y(t)=F(V(t),t)$
{\small
\begin{eqnarray*}
\lefteqn{dF(V(t),t)} &&
\\ 
&& =F_t(V(t),t)dt+F_v\left[(\alpha V(t)V(t-L_{1}) -C)dt+g(V(t-L_{2}))V(t)dW(t)\right]\\
&& +\dfrac{1}{2}F_{vv}\left[g^2(V(t-L_{2}))V^2(t)\right](dW(t))^2.
\end{eqnarray*}}
Hence the dynamic equation for $Y(t)$ is
{\small
\begin{align}
\label{eq61}
\left\lbrace \begin{array}{l}
dY(t)=\left[\dfrac{1}{2}F_{vv}g^2(V(t-L_{2}))V^2(t)+F_v(\alpha V(t)V(t-L_{1}) -C)+F_t(V(t),t)\right]dt\\
\qquad \qquad +F_{v}g(V(t-L_{2}))V(t)dW(t)
\end{array} \right.
\end{align}}
From the uniqueness of solution to stochastic delay differential equation  in \cite[Theorem 1]{GD4}, 
we have the equality almost surely of the coefficients of the corresponding terms $dt$ and $dW(t)$ in (\ref{eq1}) and (\ref{eq61}) 
as follow
{\small
\begin{eqnarray*}\label{eq:3001a}
\left\lbrace \begin{array}{l}
\alpha_yY(t)-C_y=\alpha_yF(Y(t),t)-C_y\\
\qquad \qquad \;\;\equiv \dfrac{1}{2}g^2(V(t-L_{2}))V^2(t)F_{vv}(V(t),t)+(\alpha V(t)V(t-L) -C)F_v+F_t
\end{array} \right.
\end{eqnarray*}
\begin{eqnarray}
g_y(Y(t-L_{2}))Y(t)&= &g_y(F(V(t-L_{2}),t))F(V(t),t)\\
&\equiv&
g(V(t-L_{2})) V(t)F_v(V(t),t) \label{eq:3001b}\\
dW_y(t)&\equiv& dW(t).\label{eq:3001c}
\end{eqnarray}}

Following the self-financing and replication strategy (\cite{GD4}), let $z_1$ be the instantaneous number 
corresponding to the amount invested in the firm, 
$z_2$ be the instantaneous number corresponding to the amount invested in the security and $z_3$ be the 
instantaneous number corresponding to the amount invested in riskless debt.
Consider $dx$ the instantaneous return to the portfolio and assume the total investment in the portfolio is zero, 
we may write $z_1+z_2+z_3=0$ and then
{\small
\begin{eqnarray}
dx&=&z_1\dfrac{dV(t)+Cdt}{V(t)}+z_2\dfrac{dY(t)+C_ydt}{Y(t)}+z_3rdt \nonumber\\
&=&\dfrac{z_1\left[(\alpha V(t)V(t-L_{1}) -C)dt+g(V(t-L_{2}))V(t)dW(t)\right]+Cz_1dt}{V(t)} \nonumber\\
& &+\dfrac{z_2\left[(\alpha_y Y(t) -C_y)dt+g_y(Y(t-L_{2}))Y(t)dW_y(t)\right]+C_yz_2dt}{Y(t)}+z_3rdt\nonumber\\
&=&z_1\alpha V(t-L_{1})dt +z_1g(V(t-L_{2}))dW(t)+z_2\alpha_ydt+z_2g_y(Y(t-L_{2}))dW_y(t)\nonumber\\
&& -(z_1+z_2)r dt. \nonumber
\end{eqnarray}}
Hence from the equivalence (\ref{eq:3001c}), we have
{\small
\begin{equation}
dx=\left[z_1(\alpha V(t-L_{1})-r) +z_2(\alpha_y-r)\right]dt+\left[z_1\,g(V(t-L_{2}))+z_2\,g_y(Y(t-L_{2}))\right]dW(t).
\end{equation}}
Since the return on the portfolio is non stochastic and there is no arbitrage condition we have: $z_1\,g(V(t-L_{2}))+z_2\,g_y(Y(t-L_{2}))=0$ and $z_1(\alpha V(t-L_{1})-r) +z_2(\alpha_y-r)=0$ 
leading to the following system:
\begin{displaymath}
\left\lbrace
\begin{array}{rl}
z_1\,g(V(t-L_{2}))+z_2\,g_y\,(Y(t-L_{2}))&=0 \\
z_1(\alpha V(t-L_{1})-r) +z_2(\alpha_y-r)&=0. \end{array}\right.
\end{displaymath}
A non trivial solution ($z_i\neq 0$) to this system exists if and only if
\begin{equation}\label{eq201}
\left(\dfrac{\alpha V(t-L_{1})-r}{g(V(t-L_{2}))}\right)=\left(\dfrac{\alpha_y-r}{g_y(F(V(t-L_{2}),t))}\right).
\end{equation}
But from (\ref{eq:3001a}) and (\ref{eq:3001b}) substituting for $\alpha_y$ and $g_y(F(V(t-L_{2}),t))$, we get
$$\alpha_y=\dfrac{\dfrac{1}{2}g^2(V(t-L_{2}))V^2(t)F_{vv}+(\alpha V(t)V(t-L_{1}) -C)F_v+F_t+C_y}{F(V(t),t)}\quad $$ and

$$\quad g_y(F(V(t-L_{2}),t))=\dfrac{g(V(t-L_{2})) V(t)F_v}{F(V(t),t)}.$$

Replacing $\alpha_y$ and $g_y(F(V(t-L_{2}),t))$ in (\ref{eq201}), we obtain
{\small
\begin{eqnarray*}
\lefteqn{\dfrac{\alpha V(t-L_{1})-r}{g(V(t-L_{2}))}}&&\\
&&=\dfrac{\dfrac{1}{2}g^2(V(t-L_{2}))V^2(t)F_{vv}+(\alpha V(t)V(t-L_{1}) -C)F_v+F_t+C_y-rF(V(t),t)}{g(V(t-L_{2})) V(t)F_v}.
\end{eqnarray*}}
By rearranging terms and simplifying, we get
{\small
\begin{eqnarray}\label{eq211}
\lefteqn{\alpha V(t)V(t-L_{1})F_v-rV(t)F_v}&&\\
&&=\dfrac{1}{2}g^2(V(t-L_{2}))V^2(t)F_{vv}+(\alpha V(t)V(t-L_{1}) -C)F_v+F_t+C_y-rF(V(t),t).\nonumber
\end{eqnarray}}
Therefore, we can rewrite equation (\ref{eq211}) as the following  random parabolic partial differential equation for $F$
$$\dfrac{1}{2}g^2(V(t-L_{2}))v^2\,F_{vv}+(r v-C)F_v+F_t+C_y-rF=0.$$
\end{proof}
For any claim whose value depends on the value of the firm and time, the equation \eqref{eq:3008} must be satisfied under some specific boundary 
conditions and initial condition. From these boundary conditions, we will be able to distinguish the debt of a firm from its equity. 
 By definition the value $V$ of the company can be written as
 \begin{equation}
V(t) = F(V(t),t)+f(V(t),t),\label{en}
 \end{equation}
 where $f(V(t),t)$ is the value of the equity, $F(V(t),t)$ the value of debt a any time $t$  before the maturity.
 Because both $F$ and $f$ can only take on non-negative values, we have that for initial condition
\begin{equation}\label{eq10}
F(0,t)=f(0,t)=0.
\end{equation}
Further $F(V(t),t)\leq V(t)$ which implies the regular condition
\begin{equation}\label{eq11}
\dfrac{F(V(t),t)}{V(t)}\leq 1.
\end{equation} 
Using relation \eqref{en}, for $F$ satisfying the RPDE \eqref{eq:3008}, the equity $f$ therefore satisfy the following equation
\begin{eqnarray}
\label{eq:3009}
 \dfrac{1}{2}g^2(V(t-L))v^{2}f_{vv}+(rv-C)f_v+f_t +C-C_y -rf=0.
\end{eqnarray}
\begin{remark}
Notice that  $C$ and $\alpha$ ($C_y$ and $\alpha_y$) in the SDDE \eqref{model} and \eqref{eq1} respectively can be time dependent functions, 
in which case they will be measurable and integrable in the interval $[0,T]$.
\end{remark}
In the accompanied paper \cite{ElisaANTO} and in \cite{ElisaThesis}, efficient numerical methods to solve \eqref{eq:3008} and \eqref{eq:3009} 
subject to final and boundary conditions \eqref{eq10}, \eqref{eq12} and \eqref{bc} are presented.

In the sequel we will assume that $C=C_{y}=0$ and provide the representations of the exact solutions.
%
\section{\uppercase{Evaluation of Debt in a Levered Firm}}
\label{VD}
In this section, we consider a claim market value as the simplest case of corporate debt and therefore use the following  assumption.
\begin{assumption}
\label{ass}
 We assume that:
\begin{itemize}
 \item[(a)] The company is financed by:
\begin{description}
\item[1.] A single class of debt
\item[2.]  The equity.
\end{description}
\item[(b)] The following restrictions and provisions are stipulated in the contract according to the bond issue
\begin{enumerate}
\item\label{res1} The firm must pay an amount $B(T)$ to the debtholders at the maturity date $T$;
\item\label{res2} In case the firm cannot make the payment, the debtholders take over the company and the equityholders lose their investment;
\item\label{res3} The firm is not allow neither to issue a new senior claim on the firm nor to pay cash dividend during the option life.
In other words, there is no coupon payment nor dividends prior to the maturity of the debt (i.e. $C=C_{y}=0$).
\end{enumerate}
\end{itemize}
\end{assumption}
 From \assref{ass}, we have the following final conditions
\begin{equation}
\label{eq12}
F(V,T)=\min[V,B(T)],\,\,\,\,\, f(V,T)=\max (V-B(T),0).
\end{equation}
Toward infty, as for option prices, we also have the following boundary conditions 
 \begin{eqnarray}
  \label{bc}
  F(V,t)\sim v-B(T)e^{-r(T-t)},\,\,\,\, f(V,t)\sim V-B(T)e^{-r(T-t)}, \text{ as }v\rightarrow\infty.
 \end{eqnarray}
Since there are no coupon payments, the values of $C_y$ and  $C$ in equation \eqref{eq:3008} are zero. Equations \eqref{eq:3008} and
\eqref{eq:3009} coupled  with final and boundary conditions are given respectively by
\begin{eqnarray}
\label{eqou}
\left \lbrace \begin{array}{l}
\dfrac{1}{2}g^2(V(t-L))v^2f_{vv}+rvf_v+f_t-rf=0,\,\, 0<t<T\\
f(v,T)=\max(v-B(T),0),\,\,\,\,\, v>0\label{eqou1}\\
f(0,t)=0, \,\,\,\, f(v,t)\sim v-B(T)e^{-r(T-t)}, \text{ as } v\rightarrow\infty, \label{eqou2}
\end{array} \right.
\end{eqnarray}
 and 
\begin{eqnarray}
\label{eqoub}
\left \lbrace \begin{array}{l}
\dfrac{1}{2}g^2(V(t-L))v^2F_{vv}+rvF_v+F_t-rF=0,\,\, 0<t<T\\
F(v,T)=\min[V,B(T)],\,\,\,\, v>0\label{eqou1b}\\
F(0,t)=0, \,\,\,\, F(v,t)\sim B(T)e^{-r(T-t)}, \text{ as }v\rightarrow\infty, \label{eqou2b}
\end{array} \right.
\end{eqnarray}
We shall solve the above parabolic partial differential equation \eqref{eqou} 
directly using the standard method based on Fourier transforms.
\begin{lemma}\label{le2}
The backward parabolic RPDE of the form \eqref{eq:3009} ( where $C= C_y=0$ ) with final solution $f(v,T)=\phi(v)$ can be transformed to 
 the well known heat equation 
\begin{equation}\label{eq:3103}
 h_\tau=h_{xx},\quad h=h(x,\tau), \tau= \tau (t), x= x(v).
\end{equation}
\end{lemma}
\begin{proof}
Let us make the following change of variables
{\small
\begin{eqnarray*}
\left\lbrace \begin{array}{l}
 v=B(T)e^x \\
 f(v,t)=B(T)e^{rt}h\left(x-\frac{1}{2}\int_t^T g^2(V(s-L_{2}))ds-r(T-t),\frac{1}{2}\int_t^T g^2(V(s-L_{2}))ds\right),
 \end{array} \right.
 \end{eqnarray*}}
 and
 \begin{eqnarray}
 \label{tranfo}
  \tau=\frac{1}{2}\int_t^Tg^2(V(s-L_2))ds.
 \end{eqnarray}

 The corresponding partial derivatives are given as $$ h_x=\dfrac{e^{-rt}}{B(T)}\dfrac{\partial v}{\partial x}f_v=\dfrac{e^{-rt}}{B(T)}vf_v,$$
{\small
\begin{eqnarray*}
f_t=B(T)re^{rt}h\left(x-\frac{1}{2}\int_t^T g^2(V(s-L_{2}))ds-r(T-t),\frac{1}{2}\int_t^T g^2(V(s-L_{2}))ds\right)\\
+B(T)e^{rt}\left[\left(\dfrac{1}{2}g^2(V(t-L_{2}))-r\right)h_x-\dfrac{1}{2}g^2(V(t-L_{2}))h_\tau\right].
\end{eqnarray*}
Then
{\small
\begin{eqnarray*}
h_\tau=\dfrac{-f_t+B(T)re^{rt}h\left(x-\frac{1}{2}\int_t^T g^2(V(s-L_{2}))ds-r(T-t),\frac{1}{2}\int_t^T g^2(V(s-L_{2}))ds\right)}{\dfrac{1}{2}B(T)e^{rt}g^2(V(t-L_{2}))}\\
\newline\\
+\dfrac{\dfrac{1}{2}B(T)e^{rt}g^2(V(t-L_{2}))h_x-rB(T)e^{rt}h_x}{\dfrac{1}{2}B(T)e^{rt}g^2(V(t-L_{2}))}.
\end{eqnarray*}}}
Applying the change of variables, we get
\begin{equation}
h_\tau=\dfrac{-f_t+rf(v,t)+\dfrac{1}{2}g^2(V(t-L_{2}))vf_v-rvf_v}{\dfrac{1}{2}B(T)e^{rt}g^2(V(t-L_{2}))},\label{eq:3100}
\end{equation}

\begin{equation}\label{eq:3101}
h_{xx}=\dfrac{e^{-rt}}{B(T)}\dfrac{\partial (vf_v)}{\partial v}\dfrac{\partial v}{\partial x}=\dfrac{e^{-rt}}{B(T)}v(f_v+vf_{vv}).
\end{equation}
Plugging \eqref{eq:3100} and \eqref{eq:3101} into \eqref{eq:3103}, we get
\begin{equation}
\dfrac{1}{2}g^2(V(t-L_{2}))v^2f{vv}+rvf_v+f_t-rf=0,
\end{equation}
We can observe from \eqref{tranfo} that the final condition  in \eqref{eqou} correspond to the initial condition in \eqref{eq:3103}.
\end{proof}

\begin{theorem}
\label{th2} Assume that the value of the firm $V(t),\,t\in[0,T]$ follows the SDDE \eqref{model}. Furthermore, suppose that \assref{A2} and \ref{ass} are satisfied. The equity function $f$, solution of \eqref{eqou} and the debt function $F$, solution of  \eqref{eqoub} are given respectively by

\begin{eqnarray}
 f(V(t),t)=V(t)\Phi(x_1)-Be^{-r(T-t)}\Phi(x_2),\label{eq17}\\
 F(V(t),t)=Be^{-r(T-t)}\left[ \Phi(x_2) + \dfrac{1}{d}\Phi(-x_1)\right],\label{eq17b}
\end{eqnarray}
where
\begin{eqnarray}
\label{expression}
\left\lbrace \begin{array}{l}
\Phi(x)=
\dfrac{1}{\sqrt{2\pi}}\int_{-\infty}^{x}e^{-\frac{1}{2}y^2}dy,\\
\newline\\
x_1=\dfrac{\log\,\dfrac{V(t)}{B}+r(T-t)+\frac{1}{2}\int_t^Tg^2(V(s-L_{2}))ds}{\sqrt{\int_t^Tg^2(V(s-L_{2}))ds}},\\
\newline\\
x_2=x_1-\sqrt{\int_t^Tg^2(V(s-L_{2}))ds}\\
\newline\\
d=\dfrac{Be^{-r(T-t)}}{V(t)}
\end{array} \right.
\end{eqnarray}
\end{theorem}
\begin{proof}
From \lemref{le2},  we have
\begin{eqnarray}\label{eqn4}
\lefteqn{f(V,t)}\\
&=&B(T)e^{rt}h(x,\tau),\nonumber\\
&=&B(T)e^{rt}h(x-\dfrac{1}{2}\int_t^Tg^2(V(s-L_{2}))ds+r(T-t),\dfrac{1}{2}\int_t^Tg^2(V(s-L_{2}))ds) \nonumber.
\end{eqnarray}

But, the fundamental solution to the diffusion equation \eqref{eq:3103} is given by the Green's function
\begin{equation*}\label{eqn5}
G(x,\tau)=\dfrac{1}{\sqrt{4\pi\tau}}e^{-\frac{x^2}{4\tau}}.
\end{equation*}
Furthermore, the general solution $h$ with initial condition $h(x,0)=\phi(x)$ is given by the convolution
\begin{eqnarray*}
h(x,\tau)&=&h(x,0)*G(x,\tau)\\\
&=&\int_{-\infty}^{\infty}G(x-\eta,\tau)\phi(\eta)d\eta\\\
&=&\int_{-\infty}^{\infty}\dfrac{1}{\sqrt{4\pi\tau}}\phi(\eta)\exp\left[-\dfrac{(x-\eta)^2}{4\tau} \right]d\eta.
\end{eqnarray*}
Now,
\begin{equation}
\begin{array}{ll}
h(x,\tau)=\dfrac{e^{-rT}}{\sqrt{4\pi\tau}}\int_{-\infty}^{\infty}\max[e^{\eta}-1,0]\exp\left[-\dfrac{\left(x-\tau-rt+rT -\eta\right)^2}{4\tau}\right] d\eta&\\
\qquad\quad=\dfrac{e^{-rT}}{\sqrt{4\pi\tau}}\int_{0}^{\infty}\exp\left[\dfrac{4\eta\tau -\left(x-\tau-rt+rT -\eta\right)^2}{4\tau}\right] d\eta&\\
\qquad\qquad\qquad-\dfrac{e^{-rT}}{\sqrt{4\pi\tau}}\int_{0}^{\infty}\exp\left[ -\dfrac{\left(x-\tau-rt+rT -\eta\right)^2}{4\tau}\right] d\eta&\\
\qquad\quad=I_1-I_2,&
\end{array}
\end{equation}
where $$I_1=\dfrac{e^{-rT}}{\sqrt{4\pi\tau}}\int_{0}^{\infty}\exp\left[\dfrac{4\eta\tau -\left(x-\tau-rt+rT -\eta\right)^2}{4\tau}\right] d\eta$$ and $$I_2=\dfrac{e^{-rT}}{\sqrt{4\pi\tau}}\int_{0}^{\infty}\exp\left[ -\dfrac{\left(x-\tau-rt+rT -\eta\right)^2}{4\tau}\right] d\eta.$$ We first solve $I_1$. We make the following change of variable $$z=\dfrac{4\eta\tau -\left(x-\tau-rt+rT -\eta\right)^2}{\sqrt{2\tau}}\rightarrow d\eta=\sqrt{2\tau}dz.$$ Completing the perfect square in the exponential of the integrand, we have
$$\begin{array}{ll}
\dfrac{4\eta\tau -\left(x-\tau-rt+rT -\eta\right)^2}{4\tau}&\\
\qquad\quad=-((x+\tau-rt+rT)-\eta)^2+(x+\tau-rt+rT)^2-(x-\tau-rt+rT)^2&\\
\qquad\quad=-((x+\tau-rt+rT)-\eta)^2+4\tau(x-rt+rT).&
\end{array}$$
  Moreover, we define the lower limit of the integration as $$x_1=\dfrac{\left(x+\tau-rt+rT\right)}{\sqrt{2\tau}}.$$ Hence we can write 
$$\begin{array}{ll}
I_1=\exp\left[-rT+\dfrac{4\tau(x-rt+rT)}{4\tau}\right]\int_{-x_1}^{\infty}\dfrac{\exp\left[-\dfrac{z}{2}\right]}{\sqrt{2\pi}} dz&\\
\quad=e^{x-rt}\Phi(x_1).&
\end{array}$$
We will compute $I_2$ in a similar way. Let us make the following change of variable $$y=\dfrac{\left(-x+\tau+rt-rT +\eta\right)}{\sqrt{2\tau}}\rightarrow\eta=\sqrt{2\tau}dy.$$ 
As before we define the lower limit of the integration as $$x_2=\dfrac{\left(x-\tau-rt+rT\right)}{\sqrt{2\tau}}$$
 $$\begin{array}{ll}
I_2=\dfrac{e^{-rT}}{\sqrt{4\pi\tau}}\int_{0}^{\infty}\exp\left[ -\dfrac{\left(x-\tau-rt+rT -\eta\right)^2}{4\tau}\right] d\eta&\\
\quad=e^{-rT}\int_{-x_2}^{\infty}\dfrac{\exp[-\frac{y^2}{2}]}{\sqrt{2\pi}}
\quad=e^{-rT}\Phi(x_2).&
\end{array}$$
Using the results above, we obtain the solution of the heat equation
\begin{equation}
h(x,\tau)=e^{x-rt}\Phi(x_1)-e^{-rT}\Phi(x_2).
\end{equation}
Now, we want to proceed backward using the relation between $h,x\text{ and }\tau$ and $f,v\text{ and }t$, respectively to get the solution for the RPDE.
From relation \eqref{eqn4}, we can write the solution of the RPDE \eqref{eqou} as follow
\begin{equation*}
\begin{array}{ll}
f(v,t)=v\Phi(x_1)-B(T)e^{-r(T-t)}\Phi(x_2).&
\end{array}
\end{equation*}
Finally, if we replace $v$ by the value of the company $V(t)$ in the above equation, we obtain the formula for the equity value
\begin{equation}
\begin{array}{ll}
f(V(t),t)=V(t)\Phi(x_1)-B(T)e^{-r(T-t)}\Phi(x_2),&
\end{array}
\end{equation}

The debt function $F$ in \eqref{eq17b} is obtained using the relation $F(V(t),t)=V(t)-f(V(t),t)$. More details can be found in \cite{ElisaThesis}.
\end{proof}
Notice that we can consider $F$ as a function of $V$ and $\tau_{1}=T-t$ (see \cite{GD2} when $g$ is constant).  Let  $R(\tau_{1})$ 
  the yield to the maturity on the risky debt provided that the firm does not default defined by
\begin{eqnarray}
 e^{-R(\tau_{1}) \tau_{1}}=\dfrac{F(V,\tau_{1})}{B}.
\end{eqnarray}
We therefore have  
\begin{equation}
 R(\tau_{1})-r=-\dfrac{1}{\tau_{1}}\log\,\left\lbrace\Phi\left[x_2\right]+\dfrac{1}{d}\Phi\left[-x_1\right]
 \right\rbrace
 \label{eq181}
\end{equation}
It seems reasonable to call $R(\tau_1)-r$ a risk premium in which case equation (\ref{eq181}) defines a risk structure of interest rates.
As in Merton (\cite{GD2}), the risk premium is a function of the volatility function $g$ and $d$.

\section{\uppercase{Evaluation of Loan Guarantees}}
Now let us examine the impact of a guarantor, that is a government or an institution insuring payment to the bondholders in any case.
Now we make the following assumption.
\begin{assumption}
\label{ass1}
We assume that
\begin{itemize}
 \item[(a)] The company is financed by:
 \begin{itemize}
  \item[1.] A single class of debt.
  \item [2.] The equity.
  \item [3.]  The guarantee on the debt.
 \end{itemize}
\item [(b)] The following restrictions and provisions are stipulated in the contract according to the loan guarantees issue.
 The contract stipulate that 
 \begin{enumerate}
\item\label{res1} In case the management on the maturity date is unable to make the payment promised, the government will meet these payments with no uncertainty;
\item\label{res2} The firm is expected to pay an amount at least equal to its actuarial cost for the guarantee, so that in case this happens, 
the firm is required to default all its assets to the guarantor;
\item\label{res3} The firm is not allow neither to issue a new senior claim on the firm nor to pay cash dividend during the option life i.e $C=C_y=0$.
\end{enumerate}
\end{itemize}
\end{assumption}

 Notice that the presence of a guarantor transforms the debt which was a risky asset to a riskless asset. If the firm value is less than the promised payment,
 then the debtholders receive $B$, the equity holders receive nothing. Therefore, the guarantor lose the amount $B - V(T)$. However, if the firm value 
 is greater than the promised payment, then the debtholders receive $B$ and the equity holders receive $V(T) - B$ as without the guarantee.  In other words,
 the guarantor has no impact on the equity value ($\max[V(T)-B,0]$) at the maturity date, but the debt value is riskless and always known as $B$. However, 
 the value of the guarantor is the non positive value $\min[V(T)-B,0]$. In effect, the result of the guarantee is to create an additional cash inflow to
 the firm of the amount $-\min[V(T)-B,0]$. But, $-\min[V(T)-B,0]=\max[B-V(T),0]$. Therefore, if $G(T)$ is the cost we are looking for, where the length 
 of time until the maturity date of the bond is $T$, we can write
\begin{eqnarray}
 G(T)=G(V(T),T)=\max[B-V(T),0],
 \end{eqnarray}
 
with 
 \begin{eqnarray}
\label{eqoug}
\dfrac{1}{2}g^2(V(t-L))v^2 G_{vv}+rv G_v+G_t-rG=0,\,\, 0<t<T\\
G(0,t)=0, 
\end{eqnarray}
where $B$ can be taken as the strike price and $V(T)$ as the stock price $S(T)$. These similarities
between the evaluation of $G(T)$ and the evaluation of an European put option allow us to say that loan guarantees works as an European put option on 
the firm value giving to the management the right but not an obligation to sell the amount $B$ to a guarantor.

\begin{theorem}\label{proa1}
Assume that the value of the firm $V(t),\,t\in[0,T]$ follows the SDDE \eqref{model}. Furthermore, suppose that \assref{A2} and \ref{ass1} are satisfied.
Then loan guarantees (a fair premium equal to the present value of the cash flows from the option) is given by
\begin{eqnarray} 
 \label{eq1821}
G(V(t),t) = Be^{-r(T-t)}\Phi \left[x_1\right]
 -V(t)\Phi\left[x_2\right], \nonumber
\end{eqnarray}
with \eqref{expression}.
\begin{proof}
 The proof  is similar to the proof of \thmref{th2} using \lemref{le2} and can be found in \cite{ElisaThesis}. 
\end{proof}

\end{theorem}
\begin{remark}
Notice that the probabilistic methods can be used to derive the the equity $f$, the debt $F$ and the loan guarantee $G$. The technique is the same as in \cite{GD4} for options price.
More details can be found in \cite{ElisaThesis,GD4}.
The analysis of the risk structure for an homogeneous class of debt is done in the same way as in Merton model in \cite{GD2}(see \cite{ElisaThesis}).

\end{remark}
\begin{remark}
Notice that when the volatility function $g$ is constant, all results in this paper are the same as Merton's results in \cite{GD2}.
\end{remark}
\section{\uppercase{Impact of an additional debt on the firm's risk structure}}
Let us verify the impact of the guarantee on the company.
Assume a levered company financed by equity and debt.
Assume the face value of the debt is $B$. Let us compute the probability
 of default of this company given by $P(V(T)<B)$.
From the work in \cite{GD4,ElisaThesis}, by setting $\tilde{V}(t)=e^{-rt}V(t)$  we have 
 \begin{eqnarray}\label{eq:2010}
\tilde{V}(t)=\varphi(0)\exp\left(\int_0^tg(V(s-L_2))dW^*(s)-\dfrac{1}{2}\int_0^tg^2(V(s-L_2))ds\right),
\end{eqnarray} 
where
 $$W^*(t) := W(t)+\int_0^t\dfrac{\alpha V(s-L_1)-r}{g(V(s-L_2))}ds,\,\, t \in[0,T],$$
Indeed,
{\small
\begin{eqnarray*}
 \lefteqn {P(V(T)<B)}\\
&=&P\left( \tilde{V}(T)e^{rT}<B\right) \\
 &=&P\left(\tilde{V}(t)\exp\left(-\frac{1}{2}\int_t^Tg^2(V(s-L_{2}))ds +\int_t^Tg(V(s-L_{2}))dW^*(s)\right)<Be^{-rT}\right)\\
&=&P\left(\exp\left(-\frac{1}{2}\int_t^Tg^2(V(s-L_{2}))ds +\int_t^Tg(V(s-L_{2}))dW^*(s)\right)<\dfrac{Be^{-r(T-t)}}{V(t)}\right)\\
&=&P\left(\int_t^Tg(V(s-L_{2}))dW^*(s)<\log\,d+\frac{1}{2}\int_t^Tg^2(V(s-L_{2}))ds\right)\\
&=&P\left(\dfrac{\int_t^Tg(V(s-L_{2}))dW^*(s)}{\sqrt{\int_t^Tg^2(V(s-L_{2}))ds}}<\dfrac{\log\,d+\frac{1}{2}\int_t^Tg^2(V(s-L_{2}))ds}{\sqrt{\int_t^Tg^2(V(s-L_{2}))ds}}\right)\\
&=&\Phi\left(\dfrac{\log\,d+\frac{1}{2}\int_t^Tg^2(V(s-L_{2}))ds}{\sqrt{\int_t^Tg^2(V(s-L_{2}))ds}}\right)\\
&=&\Phi\left(x_{1}\right)
\end{eqnarray*}}

Notice that we have used  the fact that $$\dfrac{\int_t^Tg(V(s-L))dW^*(s)}{\sqrt{\int_t^Tg^2(V(s-L))ds}}$$ 
is normally distributed with mean $0$ variance $1$ so that $\int_t^Tg(V(s-L))dW^*(s)$ is normally distributed 
with mean $0$ and variance $\int_t^Tg^2(V(s-L))ds$.

Now let us consider that a different debt is added to value of the company $V(t)$ and compare the probability of default with the previous.
An additional debt of face value $B^{'}$ will increase the total face value which becomes $B+B^{'}$.
 If $V(t)>B$, that is $\dfrac{V(t)}{B}>\dfrac{V(t)+B^{'}}{B+B^{'}}=\dfrac{V^{'}(t)}{B+B^{'}}$ where $V^{'}(t)=V(t)+B^{'}$. Since logarithm is an increasing function, we can write
$\log\left( \dfrac{B}{V(t)}\right) <\log\left( \dfrac{B+B^{'}}{V^{'}(t)}\right)$. So that, $x_{1}<x^{'}_{1}$ where
$$x^{'}_{1}=\dfrac{\log\left( \dfrac{B+B^{'}}{V^{'}(t)}\right)-r(T-t)+\frac{1}{2}\int_t^Tg^2(V(s-L))ds}{\sqrt{\int_t^Tg^2(V(s-L))ds}}$$ and therefore
$\Phi(x_{1})<\Phi(x_{1}^{'}).$
From the previous analysis, we can say that loan guarantees do not prevent bankruptcy. They mainly care about debtholders investments.

Now, suppose $V(t)<B$ then $\Phi(x_{1})>\Phi(x_{1}^{'}).$ This means, if the firm is already less than the face value of the debt, they may be a chance that
an additional debt may decrease the probability of default. But the question is will this additional debt able to avoid firm to a new bankruptcy situation? 
For this reason, we need to compute what is the profitability index for a new project that we want to invest in and decide.
\section*{ACKNOWLEDGEMENTS}
Antoine Tambue  is supported  by the Research Council
of Norway  under grant number 190761/S60.


\begin{thebibliography}{00}
%
\bibitem{GD4}  Arriojas, M., \ Hu, Y., Mohammed,\ S., and  Pap, G. 2007.
{\em A Delayed Black and Scholes Formula},
Journal of Stochastic Analysis and Applications, {\bf 25 (2)}, 471--492.
 \bibitem{ito}
 It\^{o}, K., and Nisio, M. 1964. 
 {\em On stationary solutions of a stochastic
differential equation.} J. Math. Kyoto University 41:1--75.
\bibitem{ElisaThesis}
Kemajou,\ E, 2012. {\em A S\lowercase{TOCHASTIC} D\lowercase{ELAY} M\lowercase{ODEL  FOR} P\lowercase{RICING} C\lowercase{ORPORATE}
 L\lowercase{IABILITIES}},
 PhD thesis, Southern Illinois University in Carbondale, USA.
\bibitem{ElisaANTO} 
Kemajou,\ E, Tambue,\ A., and  Mohammed, S.-E. A. 2012.
 { \em A  Stochastic Delay Model for Pricing Debt and Loan Guarantees:
Numerical techniques and Applications}, In preparation.
\bibitem{GD2}  Merton,\ R.\ C.\ 1974.
{\em On the Pricing of Corporate Debt: The Risk Structure of Interest Rates},
Journal of Finance, {\bf 29 (2)}, 449--470.
\bibitem{GD3} Merton,\ R.\ C.\ 1977.
{\em An Analytic Derivation of the Cost of Deposit Insurance and Loan Guarantees},
Journal of Banking and Finance, {\bf 9}, 3--11.
\bibitem{GD7}  Merton, R.\ C.\ 1973.   {\em Continuous-Time Speculative Processes: Appendix to Paul A. Samuelson's 'Mathematics of Speculative Price'}
SIAM Review,  15:34--38.
\bibitem{Mi}
Mizel, V.J., and Trutzer, V. 1984.
{ \em Stochastic hereditary equations: existence
and asymptotic stability.}  Journal of Integral Equations 7:1--72.
\bibitem{Ma}
Mohammed, S.-E.A. 1998. 
{\em Stochastic differential systems with memory:
Theory, examples and applications. } In:''Stochastic Analysis''. Decreusefond
L., Gjerde J., \O{}ksendal B., Ustunel A.S. (Eds.)  Progress in Probability 42,
Birkhauser, 1--77.
\bibitem{MS} Mohammed, S.-E. A. 1984.
{\em Stochastic Functional Differential Equations},
Pitman 99, Boston-London-Melbourne.
\bibitem{RU}
Rubinstein, M. 1994. { \em Implied binomial trees.} Journal of Finance
49(2):711--818.
\bibitem{scott}
Scott, L. 1987. { \em Option pricing when the variance changes randomly: theory,
estimation and an application.} Journal of Financial and Quantitative Analysis
22(4):419–438.

 \bibitem{Va}
 Various Authors, 2009.
 { \em Qfinance: the ultimate resource /Qatar Financial Centre,}  Bloomsbury.

%
%

\end{thebibliography}
\end{document}